%% file: main.tex
\documentclass{article}
\usepackage[preprint]{neurips_2020}
\input{preamble.tex}

\title{Towards practical differentially private causal graph discovery}

\begin{document}

\input{macro.tex}

\maketitle

\vspace{-10pt}
\begin{abstract}
Causal graph discovery refers to the process of discovering causal relation graphs from purely observational data. 
Like other statistical data, a causal graph might leak sensitive information about participants in the dataset.
In this paper, we present a differentially private causal graph discovery algorithm, \alg, which improves both utility and running time compared to the state-of-the-art.
The design of \alg follows a novel paradigm called \sieve which uses a small amount of privacy budget to filter out ``insignificant'' queries, and leverages the remaining budget to obtain highly accurate answers for the ``significant'' queries.
We also conducted the first sensitivity analysis for conditional independence tests including conditional Kendall's $\tau$ and conditional Spearman's $\rho$.
We 
evaluated \alg on \countofdata public datasets and compared with the state-of-the-art. 
The results show that \alg achieves 10.61 to 32.85 times speedup and better utility. 
%
\end{abstract}

\input{introduction}

\input{preliminaries}
\input{dpci}
\input{evaluation}
\input{related_work}
\input{conclusion}
\input{impact}

\bibliographystyle{plain}
\bibliography{ref.bib}

\appendix
\input{pc_code}
\input{em_svt}
\input{error_bound}
\input{kendall_sensitivity}
\input{spearman_sensitivity}
\input{aggregate}

\end{document}

%% file: preamble.tex
\usepackage[utf8]{inputenc}
\usepackage{color}
\usepackage[noend,linesnumbered,ruled,vlined]{algorithm2e}
\usepackage{amsmath}
\usepackage{amsthm}
\usepackage{amsfonts}
\usepackage{bbm}
\usepackage[english]{babel}
\usepackage{amssymb}
\usepackage{pifont}
\usepackage{graphicx}
\usepackage{multirow}
\usepackage{multicol}
\usepackage{pgfplots}
\pgfplotsset{compat=1.14}
\usepackage{caption}
\usepackage{subcaption}
\usepackage{xspace}
\usepackage{enumitem}
\usepackage{tikz}
\usetikzlibrary{fit,calc}

%% file: macro.tex
\newcommand{\lun}[1]{\textcolor{red}{lun:#1}}
\newcommand{\dawn}[1]{\textcolor{blue}{dawn:#1}}
\newtheorem{thm}{Theorem}
\newtheorem{lemma}{Lemma}
\newtheorem{definition}{Definition}
\newcommand{\cmark}{\ding{51}}%
\newcommand{\xmark}{\ding{55}}%
\newcommand{\alg}{\texttt{Priv-PC}\xspace}
\newcommand{\empc}{\texttt{EM-PC}\xspace}
\newcommand{\svt}{\texttt{SVT}\xspace}
\newcommand{\asvt}{\texttt{ASVT}\xspace}
\newcommand{\myem}{\texttt{EM}\xspace}
\newcommand{\pc}{\texttt{PC}\xspace}
\newcommand{\aat}{\texttt{AAT}\xspace}
\newcommand{\probe}{\texttt{Probe-and-examine}\xspace}
\newcommand{\linesofcode}{245\xspace}
\newcommand{\countofdata}{4\xspace}
\newcommand{\ubspeedup}{XX\xspace}
\newcommand{\lbspeedup}{XX\xspace}
\newcommand{\Sieve}{\texttt{Sieve-and-examine}}
\newcommand{\sieve}{\texttt{sieve-and-examine}\xspace}
\newcommand{\svtpc}{\texttt{SVT-PC}\xspace}
\newcommand{\code}{\url{XX}\xspace}

\newcommand{\tikzmark}[1]{%
  \tikz[overlay,remember picture,baseline] \node [anchor=base] (#1) {};}

\definecolor{mybrown}{RGB}{255,218,195}
\definecolor{myframe}{RGB}{197,122,195}

%% file: introduction.tex
\vspace{-10pt}
\section{Introduction}
\vspace{-5pt}

Causal graph discovery refers to the process of discovering causal relation graphs from purely observational data.
Causal graph discovery has seen wide deployment in areas like genomics, ecology, epidemiology, space physics, clinical medicine, and neuroscience. 
The \pc algorithm~\cite{spirtes2000causation} is one of the most popular causal discovery algorithms.
It is comprised of a series of independence tests like Spearman's $\rho$~\cite{spearman1961proof}, Kendall's $\tau$~\cite{kendall1938new}, G-test~\cite{mcdonald2009handbook} or $\chi^2$-test~\cite{mchugh2013chi}.
The algorithm starts by connecting all variables in the graph.
If an independence test indicates that two variables are independent, the edge between the two variables will be removed from the causal graph.
The process will continue until the edges between independent variables are totally removed.

Like other statistical data, a causal graph can leak information about participants in the dataset.
For instance, Genome-Wide Association Studies involve finding causal relations between Single Nucleotide Polymorphisms (SNPs) and diseases.
In this case, a causal link between a specific SNP and a disease may indicate the participation of a minority patient.
However, the problem of \emph{effective causal graph discovery with  differential privacy} remains largely unsolved.

\vspace{-5pt}
\paragraph{State-of-the-art.} The most straightforward solution is to perturb all the independence tests in the \pc algorithm with calibrated noise such as Laplace or Gaussian noise~\cite{dwork2014algorithmic}.
However, as pointed out in \cite{xu2017differential}, the numerous independence tests incur \emph{too much noise to output meaningful causal graphs}. 
Even tight composition techniques based on R\'enyi differential privacy~\cite{abadi2016deep, mironov2017renyi, wang2018subsampled} cannot address the issue.
The state-of-the-art solution to differentially private causal graph discovery is \empc~\cite{xu2017differential}, a modification of the \pc algorithm which uses the exponential mechanism to guarantee differential privacy.
%
%
Instead of perturbing each independence test with noise, \empc randomly selects how many and which edges to delete using the exponential mechanism. 
In this way, \empc manages to achieve a relative balance between utility and privacy.
However, \empc has two severe defects.
First, \empc suffers from \emph{extremely slow computation} because:
1) many independence tests which should have been pruned have to be computed because the exponential mechanism can only deal with off-line queries;
2) the utility function used in applying the exponential mechanism is computationally intensive.
In fact, the computation overhead of the utility score is so large that the implementation from the original paper~\cite{xu2017differential} uses a greedy search to approximate the solution presented in the paper.
It is unclear whether the differential privacy still holds given this compromise.
Second, \empc also suffers from low utility because it changes the intrinsic workflow of the \pc algorithm.
Concretely, \empc explicitly decides how many edges to delete while PC makes this decision in an on-line fashion.
Thus, \empc does not converge to the \pc algorithm and cannot achieve perfect accuracy even with substantial privacy budget.

\vspace{-7pt}
\paragraph{Proposed solution.} 
In this paper, we proposed \alg, a differentially private causal graph discovery algorithm with \emph{much less running time} and \emph{better result utility} compared to \empc.
The design of \alg follows a novel paradigm called \sieve.
Intuitively, \sieve spends a small amount of privacy budget to filter out ``insignificant'' queries and answers the rest of queries carefully with substantial privacy budget.
The proof that \alg is differentially private is straightforward.
The challenge is to understand why it also gives less running time and better utility.

\Sieve, as the name indicates, comprises two sub-processes executing alternately: \texttt{sieve} and \texttt{examine}.
In the context of causal graph discovery, the \texttt{sieve} process uses sub-sampled sparse vector technique~\cite{dwork2014algorithmic, balle2018privacy} to filter out variable pairs unlikely to be independent with a little privacy budget.
Then the \texttt{examine} process uses Laplace mechanism~\cite{dwork2014algorithmic} to carefully check the remaining variable pairs and decide whether they are really independent with substantial privacy budget.
%

We choose sparse vector technique for its nice properties.
First, sparse vector technique can answer a large number of threshold queries but only pay privacy cost for those whose output is above the threshold\footnote{Sparse vector technique can also only pay for queries below the threshold. For clarity, we only focus on the above-threshold queries throughout the paper.}.
Fortunately, in causal graph discovery, only a few independence tests will yield results above the threshold so with sparse vector technique, we can save much privacy cost.
Second, sparse vector technique can deal with online queries, so redundant independence tests can be pruned adaptively once their target edge is removed due to a previous independence test.
Thus, with sparse vector technique, we can get rid of the unnecessary independence tests in \empc and significantly accelerate private causal discovery.
We propose to further accelerate the execution and reduce privacy cost by augmenting the sparse vector technique using sub-sampling without replacement~\cite{balle2018privacy}.

However, sparse vector technique is known for its poor utility~\cite{lyu2016understanding}, which raises concern about the accuracy of \sieve.
Actually, there exist two types of errors in \sieve.
Type I error refers to mistakenly filtering out truly independent pairs.
Type II error refers to the failure to filter out variable pairs that are not independent.
To reduce the errors, we take a two-step approach.
First, we suppress type I error by tweaking the threshold lower so the noise is more unlikely to flip over the output from independence to the opposite.
The tweak, on the other hand, will increase the number of type II errors.
Fortunately, type II errors can be corrected by the \texttt{examine} process with a high probability.
Furthermore, the threshold tweak typically only increases type II errors slightly because a meaningful threshold should be far away from the clusters of both independent pairs and dependent pairs.

\vspace{-5pt}
\paragraph{Independence tests in Priv-PC.}
The noise magnitude in \alg grows proportionally to the sensitivity of the independence test (Section~\ref{sec:dp}).
Thus, to obtain an appropriate noise level, we conducted rigorous sensitivity analysis for commonly used conditional independence tests including conditional Kendall's $\tau$~\cite{korn1984kendall, taylor1987kendall} and conditional Spearman's $\rho$~\cite{taylor1987kendall} (Appendix~\ref{sec:tau}, \ref{sec:rho}).
%
We finally chose Kendall's $\tau$ in \alg because of its small sensitivity.
It also remains an interesting open question how to integrate independence tests with infinite sensitivity such as G-test~\cite{mcdonald2009handbook} or $\chi^2$-test~\cite{mchugh2013chi} in \alg (Appendix~\ref{sec:aggregate}).

%% file: preliminaries.tex
\vspace{-5pt}
\section{Preliminaries}
\vspace{-5pt}

In this section, we review necessary background knowledge about differential privacy and causal graph discovery. 

\subsection{Differential Privacy}
\label{sec:dp}

Differential privacy, formally introduced by Dwork et al.~\cite{dwork2006calibrating} has seen rapid development during the past decade and is accepted as the golden standard for private analysis.
%

\begin{definition} [($\epsilon, \delta$)-differential privacy]
A (randomized) algorithm $\mathcal{A}$ with input domain $D$ and output range $\mathcal{R}$ is ($\epsilon, \delta$)-differentially private if $\forall$ neighboring datasets $\mathcal{D}, \mathcal{D}'\in D$, and $\forall \mathcal{S} \subseteq \mathcal{R}$, we have that:
\begin{equation*}
    \mathbb{P}[\mathcal{A}(\mathcal{D})\in \mathcal{S}] \leq e^{\epsilon}\mathbb{P}[\mathcal{A}(\mathcal{D}')\in \mathcal{S}] + \delta
\end{equation*}
If $\delta=0$, it is called $\epsilon$-differential privacy or pure differential privacy.
\end{definition}

Intuitively, the definition requires a differentially private algorithm to produce similar outputs on similar inputs. 
A common approach to achieving differential privacy is to perturb the output with noise. 
The noise is carefully calibrated to appropriately mask the maximum difference of the output defined as sensitivity.

\begin{definition}[$\ell_k$-sensitivity] The $\ell_k$-sensitivity of a function $f: D\rightarrow \mathcal{R}$ is:
\begin{equation*}
\Delta f = \max_{x,y\in\mathcal{D}, \|x-y\|=1} \|f(x)-f(y)\|_k 
\end{equation*}
\end{definition}

Since all the independence tests in this paper output scalars, we omit the used norm and refer to the value as sensitivity uniformly.

Composability is an important property of differential privacy.
If several mechanisms are differentially private, so is their composition.
The privacy parameters of the composed mechanism can be derived using standard composition theorem like advanced composition~\cite{dwork2014algorithmic} and moments accountant~\cite{abadi2016deep}.
The sparse vector technique~\cite{dwork2014algorithmic} can be viewed as a special case for composition because it can answer a large number of threshold queries while only paying privacy cost for queries above the threshold.
We refer the interested readers to Appendix~\ref{sec:em_svt} for more details.

\subsection{Causal Graph Discovery}

In statistics, causal graphs are \emph{directed acyclic graphs} (DAGs) used to encode assumptions about the data-generating process, which are formally defined as follows.
\begin{definition}[Causal Graph] A causal graph $\mathcal{G}$ is a directed acyclic graph (DAG) represented by a vertex set $V=\{v_1, v_2, \cdots, v_k\}$ and an edge set $E \subseteq V\times V$. $Adj(\mathcal{G}, v_i)$ represents the adjacent set of node $v_i$ in graph $\mathcal{G}$. The skeleton of a DAG is the undirected version of the graph.
\end{definition}
Causal graph discovery refers to the process of discovering causal graphs under an observed distribution such as a dataset.
The output of a causal graph discovery algorithm is a \emph{completed, partially directed acyclic graph} (CPDAG) because the directions of some edges cannot be determined only based on the observational distribution.

%
%
There exist a variety of causal graph discovery algorithms and the \pc algorithm is one of the most popular ones.
The first step in the \pc algorithm is to find the skeleton of the causal graph using conditional independence tests. 
Then the edges are directed based on some auxiliary information from the independence tests to obtain CPDAG.
Because the second step does not touch the data, we only focus on the first step given the post-processing theorem~\cite{dwork2014algorithmic} in differential privacy.
The details of the \pc algorithm is introduced in Section~\ref{sec:privpc} and Appendix~\ref{sec:pc_code}.


\subsection{Conditional Independence Test}

Conditional independence test is an important building block in many causal discovery algorithms.
It is used to test whether two random variables are independent conditional on another set of variables.

\begin{definition}[Conditional independence test]
A conditional independence test $f: V\times V\times 2^V\times D\rightarrow \{0, 1\}$ decides whether variable $i\neq j \in V$ are independent conditional on another set of variables $k\subseteq V, i, j \notin k$. $f$ is composed of a dependence score $s: V\times V\times 2^V\times D \rightarrow \mathbb{R}$ and a threshold $T\in \mathbb{R}$.
\begin{equation*}
  f(\mathcal{D}) =
    \begin{cases}
      0, & s(\mathcal{D})\leq T\\
      1, & s(\mathcal{D}) > T
    \end{cases}  
\end{equation*},
where $1$ represents ``independent'' and $0$ represents ``not independent''.
$f$ is called $|k|$-order conditional independence test where $|k|$ is the size of the conditional set.
\end{definition}

Commonly used independence tests include Spearman's $\rho$, Kendall's $\tau$, G-test and $\chi^2$-test.
Note that some independence tests like Kendall's $\tau$ output $1$ when the dependence score is below the threshold.
However, for clarity, we assume all the independence tests output $1$ when the dependence score is above the threshold without loss of generality.
In this paper, we focus on Kendall's $\tau$ because of its small sensitivity (Section~\ref{sec:idpt}).

%% file: dpci.tex
\vspace{-5pt}
\section{Differentially Private Causal Graph Discovery}
\vspace{-5pt}

In this section, we proposed \alg to effectively discover causal graphs following \sieve paradigm.
Concretely, \alg leverages the \texttt{sieve} process to sift out variable pairs unlikely to independent using a little privacy cost and then carefully \texttt{examine}s the remaining ones with substantial privacy budget.
We first introduce \sieve mechanism and then demonstrate how to apply \sieve to the \pc algorithm to obtain \alg.
At last, we bridge \sieve and \alg by providing sensitivity analysis for Kendall's $\tau$.

\vspace{-5pt}
\subsection{Sieve-and-examine Mechanism}
\label{sec:probe}
\vspace{-5pt}


Most causal graph discovery algorithms like the \pc algorithm need to answer many independence tests --  \emph{too many to obtain an acceptable privacy guarantee} using independent perturbation mechanisms like Laplace mechanism~\cite{dwork2014algorithmic}.
\empc is the first step towards reconciling the contradiction between utility and privacy in private causal discovery.
However, \empc suffers from \emph{extremely slow running time} because it additionally runs a large number of independence tests that should have been pruned.
A straightforward solution is to replace the exponential mechanism~\cite{dwork2014algorithmic} with the sparse vector technique~\cite{dwork2014algorithmic, lyu2016understanding}.
Sparse vector technique allows adaptive queries so unnecessary independence tests can be pruned early. 
Besides, the privacy cost of sparse vector technique only degrades with the number of queries above the threshold.
Fortunately, only a few independence tests in causal discovery yield values above the threshold so the sparse vector technique can also save considerable privacy budget in causal discovery. 
However, sparse vector technique suffers from \emph{low accuracy} as pointed out in~\cite{lyu2016understanding}, which is not acceptable in many use cases such as medical or financial analysis.

To address the issue, we propose a novel paradigm called \sieve which alternately executes sub-sampled sparse vector technique and output perturbation.
Intuitively, the \texttt{sieve} process uses sub-sampled sparse vector technique to filter out independence tests unlikely to be above the threshold with small privacy budget.
Then the left queries are \texttt{examined} carefully with substantial privacy budget using output perturbation.

\paragraph{One-off
sieve-and-examine.}

For simplicity, we first introduce one-off \sieve shown in Algorithm~\ref{alg:onesieve}, a simplified version of \sieve that halts after seeing one query above the threshold. 
We prove that one-off \sieve is $\epsilon$-differentially private.
The result can be generalized to multiple above-threshold queries using composition theorem.

\begin{algorithm}
\SetKwInOut{Input}{Input}
\SetKwInOut{Output}{Output}
\SetKwProg{Fn}{Function}{:}{\KwRet}
\SetAlgoLined
\KwIn{$\mathcal{D}$: dataset, $\{f_i\}$: queries, $T$: threshold, $t$: threshold tweak, $m$: subset size, $\epsilon$: privacy parameters, $\Delta$: sensitivity of $f$}.
\Fn{Sieve\_and\_examine($\mathcal{D},\{f_i\},T,t,m,\epsilon, \Delta$)}{
    $\mathcal{D}' \overset{\$}{\leftarrow} \mathcal{D}, n=|\mathcal{D}|, m=|\mathcal{D}'|$\;
    Let $\epsilon' = \ln(\frac{n}{m}(e^{\epsilon/2}-1)+1)$\;
    Let $\hat{T}=T-t+Lap(\frac{2\Delta}{\epsilon'})$\;
    \For{Each query $i$}{
        \If{$f_i(\mathcal{D}')+Lap(\frac{4\Delta}{\epsilon'})\geq\hat{T}$}{
            Let $k=i$\;
            Break\;
        }
    }
    \lIf
    {$f_k(\mathcal{D})+Lap(\frac{2\Delta}{\epsilon})\geq T$}{
        Output k
    } \lElse{Output $\bot$}
}
\caption{One-off sieve-and-examine mechanism.}
\label{alg:onesieve}
\end{algorithm}

\begin{thm}
Algorithm \ref{alg:onesieve} is $\epsilon$-differentially private.
\label{thm:dpaat}
\end{thm}

\begin{proof}[Proof Sketch]
We separately prove that \texttt{sieve} and \texttt{examine} are both $\epsilon/2$-differentially private.
The main body of \texttt{sieve} is a sparse vector technique with $\epsilon'=\ln(\frac{n}{m}(e^{\epsilon/2}-1)+1)$ privacy cost.
Sub-sampling reduces the cost to $\epsilon/2$ following Theorem 9 from~\cite{balle2018privacy}. 
\texttt{Examine} process is a $\epsilon/2$-differentially private Laplace mechanism.
Thus, \sieve is $\epsilon$-differentially private using composition theorem.
\end{proof}

\paragraph{Result Utility.} 
The differential privacy proof is straightforward.
The challenge will be to understand when it also gives utility.
Thus, we bound the probability of type I error and type II error in Algorithm~\ref{alg:onesieve} separately and provide the proof in Appendix~\ref{sec:error}.

\begin{thm}[Error bound]~

\begin{itemize}
    \setlength\itemsep{-1em}
    \item (Type I error) Let $E_1^\alpha$ denotes the event that Algorithm~\ref{alg:onesieve} filters out $f(\mathcal{D})\geq T+\alpha$.
    $$\mathbb{P}[E_1^\alpha] \leq \exp(-\frac{\epsilon'(\alpha+t)}{6\Delta}) - \frac{1}{4}\exp(-\frac{\epsilon'(\alpha+t)}{3\Delta})$$.

    \item (Type II error) Let $E_2^\alpha$ denotes the event that Algorithm~\ref{alg:onesieve} fails to filter out $f(\mathcal{D})\leq T-\alpha$. 
    If $\alpha \geq t$, then
    $$\mathbb{P}[E_2^\alpha]\leq \exp(-\frac{12\epsilon\alpha+\epsilon'(\alpha-t)}{6\Delta}) - \frac{1}{4}\exp(-\frac{6\epsilon\alpha+\epsilon'(\alpha-t)}{3\Delta})$$.
\end{itemize}
\label{thm:error}
\vspace{-10pt}
\end{thm}

Intuitively, theorem~\ref{thm:error} bounds the probability of errors conditional on the distance from the dependence score to the threshold.
An interesting observation is the tweak on the threshold $t$ decreases the probability of type I errors and increases the probability of type II errors at the same time.
Because each type II error increases the privacy cost by $\epsilon$, the question is ``\emph{will the increment of type II errors add too much privacy cost?}''
Fortunately, the answer is ``no'' because the increment of type II errors also depends on the distribution of dependence scores.
Generally the empirical distribution of an independence score is a twin-peak curve and the threshold locates in the middle valley.
In this case, the threshold tweak only slightly increases the number of type II errors because most dependence scores are far from the threshold\footnote{A complete explanation contains two parts. First, since most of the dependence scores are far from the threshold, the threshold tweak does not directly change the test results for most queries. Second, because the dependence scores are far from the threshold, the absolute increase of type II error probability is small. Thus, the increment of type II errors is small.}.

\subsection{Priv-PC Algorithm}
\label{sec:privpc}

In this section, we demonstrate how to apply \sieve to \pc algorithm to obtain \alg. 
We first give an overview of \alg.
Then we discuss how to optimize the sub-sampling rate in \alg.


\paragraph{\alg algorithm.}


The complete pseudo-code for \alg is shown in Algorithm~\ref{alg:pc}.
\alg follows the same workflow as the \pc algorithm.
It starts from a complete undirected graph (line 1) and gradually increases the order of the independence tests (line 6, 17).
Within a fixed order, \alg traverse all the variable pairs with large enough adjacent set (line 8).
It selects the conditional variables from the adjacent set (line 9-10) and then executes the conditional independence test to decide whether the edge will be removed from the graph.

To achieve differential privacy, the conditional independence tests are augmented with \sieve.
Concretely, \alg first sub-samples a subset $\mathcal{D}'$ from $\mathcal{D}$, derives privacy parameter for the \texttt{sieve} process and tweaks the threshold (line 3-5).
Then, \alg executes the \texttt{sieve} process by adding noise to both the tweaked threshold (line 5) and the independence test (line 11). 
Note that the noise parameters here are different from standard \sieve (Algorithm~\ref{alg:onesieve}) because the sensitivity for Kendall's $\tau$ is dependent on the dataset size (Section~\ref{sec:idpt}).
Once an independence test on the sub-sampled dataset exceeds the threshold (line 11), the \texttt{examine} process will run the independence test again on the complete dataset with substantial privacy budget.
If the result still exceeds the un-tweaked threshold (line 12), the edge is removed from the graph (line 13).
Then, the sub-sampled dataset and the threshold are refreshed for the next round of \sieve (line 14-15).



\input{privpc_code}

\vspace{-5pt}
\paragraph{Optimize sub-sampling rate in \alg.}
In Algorithm~\ref{alg:pc}, we require the caller of the function to explicitly give the size of the sub-sampling set.
However, since the sensitivity of Kendall's $\tau$ also depends on the data size (Section~\ref{sec:idpt}), we can actually derive an optimal sub-sampling size which adds the smallest noise under the same privacy guarantee. 
This requires to minimize the noise level $\frac{\sqrt{n/m}}{\ln(\frac{n}{m}(\exp(\epsilon/2)-1)+1)}$.
Although there is no explicit solution for the optimization problem, we can obtain an approximate solution with numerical solver such as BFGS~\cite{nocedal2006numerical}.
On the other hand, when $\epsilon$ is small, the optimal sub-sampling size is also too small to yield meaningful independence test results.
Thus we take the optimal sub-sampling size by clipping the solution to range $(\frac{n}{20}, n)$.

\vspace{-5pt}
\subsection{Independence Tests in Priv-PC}
\label{sec:idpt}

The last missing piece is the sensitivity of the conditional independence test functions.
We finally choose conditional Kendall's $\tau$ for its small sensitivity.
Conditional Spearman's $\rho$ is another candidate but it can only be used on large datasets because of the large coefficient in its sensitivity (Appendix~\ref{sec:rho}).

The sensitivity of Kendall's $\tau$ is inversely proportional\footnote{Note that this requires the size of the dataset to be public which is a common case.} to the training set size as pointed out in~\cite{kusner2015private}. 
However, in our scenario, the conditional version of Kendall's $\tau$ is needed while \cite{kusner2015private} only gives the sensitivity for non-conditional Kendall's $\tau$. 
In order to fill the gap, we derive the sensitivity of the conditional Kendall's $\tau$, and leave the proof to Appendix~\ref{sec:tau}.

\begin{thm}
(Sensitivity of conditional Kendall's $\tau$.) The sensitivity of conditional Kendall's $\tau$ in Definition~\ref{def:ctau} (Appendix~\ref{sec:tau}) is $\frac{c_1}{\sqrt{n}}$ where $n$ is the size of the input dataset and $c_1$ is an explicit constant approaching $9/2$ when the dataset size grows.
\label{thm:stau}
\end{thm}




%% file: privpc_code.tex
\vspace{-10pt}
\begin{frame}

\begin{algorithm}
\DontPrintSemicolon
\SetKwInOut{Input}{Input}
\SetKwInOut{Output}{Output}
\SetKwProg{Fn}{Function}{:}{\KwRet}
\SetAlgoLined
\KwIn{$V$: vertex set, $\mathcal{D}$: dataset, $T$: threshold, $t$: threshold tweak, $m$: subset size, $\epsilon$: privacy parameter, $\Delta$: sensitivity on the full dataset.}
\Fn{Priv\_PC($V, \mathcal{D},T,t,m,\epsilon, \Delta$)}{
$\mathcal{G}$ = complete graph on V, $ord$ = 0\;
\tikzmark{a}$\mathcal{D}' \overset{\$}{\leftarrow} \mathcal{D}, n=|\mathcal{D}|, m=|\mathcal{D}'|$\; 
Let $\epsilon' = \ln(\frac{n}{m}(e^{\epsilon/2}-1)+1)$\;
Let $\hat{T} = T-t+Lap(\frac{2\sqrt{n}\Delta}{\sqrt{m}\epsilon'})$\tikzmark{b}\;
\While{$\exists~v_i~s.t.~ |Adj(\mathcal{G}, v_i)- v_j|\geq \text{ord}$}{
    \While{$\exists$ edge $(v_i, v_j)~s.t.~|Adj(\mathcal{G}, v_i)-v_j| \geq ord$ that has not been tested}{
        select edge $(v_i, v_j)$ in $\mathcal{G}~s.t.~|Adj(\mathcal{G}, v_i)-v_j| \geq ord$\;
        \While{$\exists~S \subseteq Adj(\mathcal{G}, v_i)-v_j$ that has not been tested}{
            choose $S \subseteq Adj(\mathcal{G}, v_i)-v_j$, $|S|=ord$\;
            \tikzmark{c}\If{$\tau(ij|S)$ $+ Lap(\frac{4\sqrt{n}\Delta}{\sqrt{m}\epsilon'}) \geq \hat{T}$}{
                \If{$\tau(ij|S) + Lap(\frac{2\sqrt{n}\Delta}{\sqrt{m}\epsilon})\geq T$}{
                    delete $(v_i, v_j)$ from $\mathcal{G}$\;
                }
                $\mathcal{D}' \overset{\$}{\leftarrow} \mathcal{D}, |\mathcal{D}'|=m$\;
                $\hat{T} = T-t+Lap(\frac{2\Delta}{\epsilon'})\tikzmark{d}$\;
                break\;
            }
        }    
    }
    $ord$ = $ord$ + 1\;}
    Output $\mathcal{G}$, compute the total privacy cost $(\epsilon_{tot}, \delta_{tot})$ with advanced composition.\;
}
\caption{\alg Algorithm with Kendall's $\tau$. The highlighted parts are different from \pc algorithm.}
\label{alg:pc}
\end{algorithm}

\begin{tikzpicture}[remember picture,overlay]
\coordinate (aa) at ($(a)+(0.1,0.3)$);
\coordinate (bb) at ($(b)+(-0.1,-0.05)$);
\node[draw=myframe,line width=1pt,fill=mybrown,opacity=0.4,rectangle,rounded corners,fit=(aa) (bb)] {};
\coordinate (cc) at ($(c)+(0.1,0.3)$);
\coordinate (dd) at ($(d)+(1.6,-0.05)$);
\node[draw=myframe,line width=1pt,fill=mybrown,opacity=0.4,rectangle,rounded corners,fit=(cc) (dd)] {};
\end{tikzpicture}
\vspace{-15pt}
\end{frame}


%% file: evaluation.tex
\vspace{-5pt}
\section{Evaluation}
\vspace{-5pt}

In this section, we evaluate the effectiveness of \alg by answering the following two questions. 1) How accurate is the result of \alg? 2) How much running time does \alg save?

\subsection{Experiment Setup}

In order to answer the above questions, we selected \countofdata datasets. 
The detailed information about the datasets is shown in Table~\ref{tab:dataset}.

\begin{table}[htbp]
    \centering
    \caption{Datasets used in the evaluation.}
    \label{tab:dataset}
    \begin{tabular}{|l|c|c|c|c|}
    \hline
    Dataset & \# Features & \# Samples & \# Edges & Type\\
    \hline
    Earthquake~\cite{korb2010bayesian} & 5 & 100K & 4 & Binary\\
    \hline
    Cancer~\cite{korb2010bayesian} & 5 & 100K & 4 & Binary\\
    \hline
    Asia~\cite{lauritzen1988local} & 8 & 100K & 10 & Binary\\
    \hline
    Survey~\cite{scutari2014bayesian} & 6 & 100K & 6 & Discrete \\
    \hline
    
    \end{tabular}
\end{table}

To directly compare \empc and \alg, we ran the two algorithms on the datasets
with 21 different privacy parameters and presented the results with accumulated privacy cost between $1$ and $100$. 
Furthermore, to demonstrate the utility improvement due to \sieve, we also directly applied sparse vector technique to \pc algorithm (\svtpc) and evaluated it under the same setting.
For each privacy parameter, we ran the three algorithms for 5 times and recorded the mean and standard deviation of the utility of the output graph and the running time.
Utility is measured in terms of $F1$-score\footnote{
If a causal graph discovery outputs $\mathcal{G}=(V, E)$ and the ground truth
is $\mathcal{G'}=(V, E')$. Then $F1$-score is defined as:
$
 F1 = \frac{2\cdot \text{Precision} \cdot \text{Recall}}{\text{Precision} + \text{Recall}},     \text{Precision} = \frac{|E\cap E'|}{|E|},
    \text{Recall} = \frac{|E\cap E'|}{|E'|}
$
}.
%
%
All the experiments were run on a Ubuntu18.04 LTS server with 32 AMD Opteron(TM) Processor 6212 with 512GB RAM.

\subsection{Utility}

\newcommand{\figwidth}{0.24\textwidth}
\tikzset{font={\fontsize{20pt}{12}\selectfont}}
\begin{figure}[htbp]
     \centering
     \begin{subfigure}[b]{\figwidth}
         \centering
         \resizebox{\textwidth}{!}{\input{figs/earthquake}}
         \caption{Earthquake.}
         \label{fig:earthquake}
     \end{subfigure}
     \hfill
     \begin{subfigure}[b]{\figwidth}
         \centering
         \resizebox{\textwidth}{!}{\input{figs/cancer}}
         \caption{Cancer.}
         \label{fig:cancer}
     \end{subfigure}
     \hfill
     \begin{subfigure}[b]{\figwidth}
         \centering
         \resizebox{\textwidth}{!}{\input{figs/asia}}
         \caption{Asia.}
         \label{fig:asia}
     \end{subfigure}
     \hfill
     \begin{subfigure}[b]{\figwidth}
         \centering
         \resizebox{\textwidth}{!}{\input{figs/survey}}
         \caption{Survey.}
         \label{fig:survey}
     \end{subfigure}
     \hfill
        \caption{$F1$-Score vs. Privacy Budget.}
        \label{fig:f1}
\vspace{-10pt}
\end{figure}
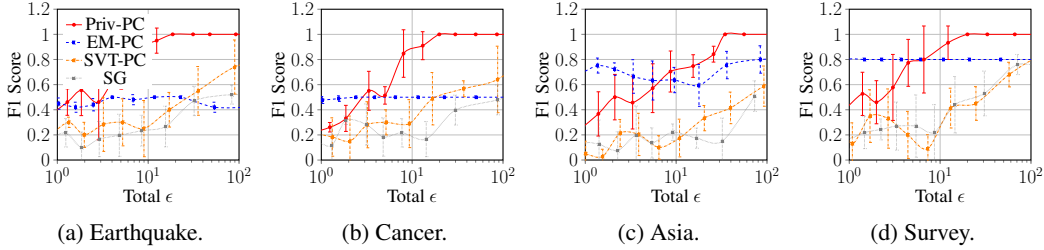

In the evaluation, \alg achieves better utility than \empc when the privacy budget is reasonably large as shown in Figure~\ref{fig:f1}. 
\alg always converges to perfect accuracy when privacy cost grows while \empc does not.
The reason is that \alg converges to \pc when privacy cost grows but \empc does not because it contains a unique sub-routine to explicitly decide the number of edges to delete.
The sub-routine intrinsically inhibits the accuracy of \empc.
On the other hand, \empc achieves better utility under small privacy budget\footnote{The line for \empc is almost flat in Figure~\ref{fig:f1} because the rising segment appears under small privacy budget out of the axis scope (approximately $0.01\sim0.1$ according to our evaluation).} because the exponential mechanism has better utility than the sparse vector technique under small privacy budget as pointed out in~\cite{lyu2016understanding}.

Compared with \svtpc, \alg always achieves much better utility (Figure~\ref{fig:f1}).
The huge improvement should be attributed to \sieve because it effectively suppresses type I and type II errors in sparse vector technique~\ref{sec:probe}.
Second, because the sensitivity of Kendall's $\tau$ is inversely proportional to the size of the input dataset, the noise is typically small when the dataset is large.
Thus, the noise does not severely harm the utility while preserving rigorous privacy.

For completeness, we also evaluated sub-sampled gaussian (\texttt{SG}) with moments accountant~\cite{abadi2016deep} as a representative for tight composition theorem~\cite{wang2018subsampled}.
The results show that the utility is typically worse than \empc and \svtpc, not to mention \alg.
%

\vspace{-5pt}
\subsection{Running Time}
\vspace{-5pt}

\begin{table}[htbp]
    \caption{Running time when privacy budget for each \sieve is $1$.}
    \label{tab:time}
    \centering
    \begin{tabular}{|c|c|c|c|c|c|c|c|c|}
    \hline
    Average Running Time & EM-PC & SVT & Priv-PC & Priv-PC w/o sub-sampling\\
    \hline
    Earthquake~\cite{korb2010bayesian} & 176.04s & 3.38s & 6.62s & 11.01s\\
    \hline
    Cancer~\cite{korb2010bayesian} & 64.62s & 2.94s & 6.09s & 10.83s\\
    \hline
    Asia~\cite{lauritzen1988local} & 531.80s & 10.06s & 16.19s & 19.40s\\
    \hline
    Survey~\cite{scutari2014bayesian} & 68.13s & 1.21s & 2.13s & 5.12s \\
    \hline
    \end{tabular}
\end{table}

\alg achieves 10.61 to 32.85 times speedup compared with \empc as shown in Table~\ref{tab:time}.
The improvement is due to two reasons.
First, \alg can deal with online queries while \empc cannot.
Thus, if an edge is removed due to a previous independence test, later tests on the same edge can be skipped to avoid extra computation overhead.
Second, in the \texttt{sieve} process, \alg only runs independence tests on a subset of the dataset which further accelerates the process.
This also explains why \alg sometimes runs faster than \svtpc.

To better understand how the two factors contribute to the speedup, we run \alg without sub-sampling under the same setting and include the results in Table~\ref{tab:time}.
The results show that the first factor provides 5.97 to 27.41 times speedup and sub-sampling provides 1.20 to 2.40 times speedup.

%% file: figs/earthquake.tex
\begin{tikzpicture}
\begin{axis}[
  grid=major,
  ymin=0, ymax=1.2, xmin=1, xmax=100, xmode=log,
  ytick align=outside, ytick pos=left,
  xtick align=outside, xtick pos=left,
  xlabel=Total $\epsilon$,
  ylabel={F1 Score},
  legend pos=north west,
  legend style={draw=none}]
\addplot+[
  red, mark options={scale=0.75},
  smooth, 
  error bars/.cd, 
    y fixed,
    y dir=both, 
    y explicit
] table [x=x, y=y,y error=error, col sep=comma] {data/earthquake/privpc.txt};
\addlegendentry{Priv-PC}
\addplot+[
  blue, mark=square*, mark options={scale=0.75},
  smooth,
  dashed,
  error bars/.cd, 
    y fixed,
    y dir=both, 
    y explicit
] table [x=x, y=y,y error=error, col sep=comma] {data/earthquake/empc.txt};
\addlegendentry{EM-PC}
\addplot+[
  orange, mark=square*, mark options={scale=0.75},
  smooth,
  densely dashdotted,
  error bars/.cd, 
    y fixed,
    y dir=both, 
    y explicit
] table [x=x, y=y,y error=error, col sep=comma] {data/earthquake/svt.txt};
\addlegendentry{SVT-PC}
\addplot+[
  gray, mark=square*, mark options={scale=0.75},
  smooth,
  densely dotted,
  error bars/.cd, 
    y fixed,
    y dir=both, 
    y explicit
] table [x=x, y=y,y error=error, col sep=comma] {data/earthquake/sub_gaussian.txt};
\addlegendentry{SG}
\end{axis}
\end{tikzpicture}

%% file: figs/cancer.tex
\begin{tikzpicture}
\begin{axis}[
  grid=major,
  ymin=0, ymax=1.2, xmin=1, xmax=100, xmode=log,
  ytick align=outside, ytick pos=left,
  xtick align=outside, xtick pos=left,
  xlabel=Total $\epsilon$,
  ylabel={F1 Score},
  legend pos=north west,
  legend style={draw=none}]
\addplot+[
  red, mark options={scale=0.75},
  smooth, 
  error bars/.cd, 
    y fixed,
    y dir=both, 
    y explicit
] table [x=x, y=y,y error=error, col sep=comma] {data/cancer/privpc.txt};
\addplot+[
  blue, mark=square*, mark options={scale=0.75},
  smooth,
  dashed,
  error bars/.cd, 
    y fixed,
    y dir=both, 
    y explicit
] table [x=x, y=y,y error=error, col sep=comma] {data/cancer/empc.txt};
\addplot+[
  orange, mark=square*, mark options={scale=0.75},
  smooth,
  densely dashdotted,
  error bars/.cd, 
    y fixed,
    y dir=both, 
    y explicit
] table [x=x, y=y,y error=error, col sep=comma] {data/cancer/svt.txt};
\addplot+[
  gray, mark=square*, mark options={scale=0.75},
  smooth,
  densely dotted,
  error bars/.cd, 
    y fixed,
    y dir=both, 
    y explicit
] table [x=x, y=y,y error=error, col sep=comma] {data/cancer/sub_gaussian.txt};
\end{axis}
\end{tikzpicture}

%% file: figs/asia.tex
\begin{tikzpicture}
\begin{axis}[
  grid=major,
  ymin=0, ymax=1.2, xmin=1, xmax=100, xmode=log,
  ytick align=outside, ytick pos=left,
  xtick align=outside, xtick pos=left,
  xlabel=Total $\epsilon$,
  ylabel={F1 Score},
  legend pos=north west,
  legend style={draw=none}]
\addplot+[
  red, mark options={scale=0.75},
  smooth, 
  error bars/.cd, 
    y fixed,
    y dir=both, 
    y explicit
] table [x=x, y=y,y error=error, col sep=comma] {data/asia/privpc.txt};
\addplot+[
  blue, mark=square*, mark options={scale=0.75},
  smooth,
  dashed,
  error bars/.cd, 
    y fixed,
    y dir=both, 
    y explicit
] table [x=x, y=y,y error=error, col sep=comma] {data/asia/empc.txt};
\addplot+[
  orange, mark=square*, mark options={scale=0.75},
  smooth,
  densely dashdotted,
  error bars/.cd, 
    y fixed,
    y dir=both, 
    y explicit
] table [x=x, y=y,y error=error, col sep=comma] {data/asia/svt.txt};
\addplot+[
  gray, mark=square*, mark options={scale=0.75},
  smooth,
  densely dotted,
  error bars/.cd, 
    y fixed,
    y dir=both, 
    y explicit
] table [x=x, y=y,y error=error, col sep=comma] {data/asia/sub_gaussian.txt};
\end{axis}
\end{tikzpicture}

%% file: figs/survey.tex
\begin{tikzpicture}
\begin{axis}[
  grid=major,
  ymin=0, ymax=1.2, xmin=1, xmax=100, xmode=log,
  ytick align=outside, ytick pos=left,
  xtick align=outside, xtick pos=left,
  xlabel=Total $\epsilon$,
  ylabel={F1 Score},
  legend pos=north west,
  legend style={draw=none}]
\addplot+[
  red, mark options={scale=0.75},
  smooth, 
  error bars/.cd, 
    y fixed,
    y dir=both, 
    y explicit
] table [x=x, y=y,y error=error, col sep=comma] {data/survey/privpc.txt};
\addplot+[
  blue, mark=square*, mark options={scale=0.75},
  smooth,
  dashed,
  error bars/.cd, 
    y fixed,
    y dir=both, 
    y explicit
] table [x=x, y=y,y error=error, col sep=comma] {data/survey/empc.txt};
\addplot+[
  orange, mark=square*, mark options={scale=0.75},
  smooth,
  densely dashdotted,
  error bars/.cd, 
    y fixed,
    y dir=both, 
    y explicit
] table [x=x, y=y,y error=error, col sep=comma] {data/survey/svt.txt};
\addplot+[
  gray, mark=square*, mark options={scale=0.75},
  smooth,
  densely dotted,
  error bars/.cd, 
    y fixed,
    y dir=both, 
    y explicit
] table [x=x, y=y,y error=error, col sep=comma] {data/survey/sub_gaussian.txt};
\end{axis}
\end{tikzpicture}

%% file: related_work.tex
\vspace{-10pt}
\section{Related Work}
\vspace{-5pt}

Causal inference has a long history and there are several excellent overviews \cite{pearl2009causal, glymour2019review} of this area. In this section, we briefly introduce the related works in the two most relevant sub-areas: causal discovery based on graph models and private causal inference.

Causal discovery based on graph models can be roughly classified into two categories. The first category is constraint-based causal discovery. The \pc algorithm~\cite{spirtes2000causation} is the most well-known algorithm in this category. It traverses all the edges and adjacent conditional sets in the causal graph and removes the edge if the conditional independence test indicates that the edge connects two independent variables. An important variation of the \pc algorithm is the Fast Causal Inference (FCI)~\cite{spirtes2000causation}, which tolerates latent confounders. The Greedy Equivalence Search (GES)~\cite{chickering2002optimal} is another widely-used algorithm in this category which starts with an empty graph and gradually adds edges. The second category is based on functional causal models (FCM). A FCM represents the effect Y as a function of the direct causes X and some noise: $Y=f(X, \epsilon; \theta)$. In Linear, Non-Gaussian and Acyclic Model (LiNGAM)~\cite{shimizu2011directlingam}, $f$ is linear, and only one of $\epsilon$ and $X$ can be Gaussian. In Post-Nonlinear Model (PNL)~\cite{zhang2006extensions, zhang2009ica}, $Y=f_2(f_1(X)+\epsilon)$. Additive Noise Model (ANM)~\cite{hoyer2009nonlinear} further constrains the post-nonlinear transformation of PNL model.

Private causal discovery has a relatively short history. 
In 2013, Johnson \emph{et al.}~\cite{johnson2013privacy} studied differentially private Genome-Wide
Association Studies (GWAS). 
They used Laplace mechanism and exponential mechanism to build specific queries of interest in GWAS.
In 2015, Kusner \emph{et al.}~\cite{kusner2015private} analyzed the sensitivity of several commonly used dependence scores on training and testing datasets, and then applied Laplace mechanism to the ANM model. 
Xu \emph{et al.}~\cite{xu2017differential} proposed to apply exponential mechanism to the \pc algorithm. 
Another line of work focuses on private bayesian inference including \cite{dimitrakakis2014robust, heikkila2017differentially, bernstein2018differentially}.
Their pioneering works are inspiring but lack novelty from differential privacy side because they all directly leverage off-the-shelf differentially private mechanisms without any modification.

%% file: conclusion.tex
\vspace{-5pt}
\section{Conclusion \& Future Work}
\vspace{-5pt}

This paper takes an important step towards practical differentially private causal discovery. 
We presented \alg, a novel differentially private causal discovery algorithm with high accuracy and short running time.
We also performed an empirical study to demonstrate the advantages compared with the state-of-the-art.
%

At the same time, we observe many challenges in differentially private causal discovery that existing techniques are
not capable of handling. 
For example, it is unclear how to reconcile independence tests with infinite sensitivity such as G-test and $\chi^2$-test; it is unclear how to handle data type beyond categorical data like numerical data. 
We consider all these problems as important future work in the research agenda toward solving the private causal discovery problem.

%% file: impact.tex
\section*{Potential Broader Impact}

\alg provides an approach to effectively discovering causal graphs from purely observational data. 
It can be deployed in genomics, ecology, epidemiology, space physics, clinical medicine, and neuroscience to release causal graph while preserving the privacy of the sensitive input data.
At the same time, \alg should be used with carefully calibrated parameters to make sure the discovered graph is accurate and preserves the privacy .
Inappropriate parameters might lead to weak privacy guarantee that does not provide strong protection against attacks such as inference attack.

%% file: pc_code.tex
\section{Pseudo-code of PC Algorithm}
\label{sec:pc_code}

For completeness, we provide the pseudo-code for the \pc algorithm in Algorithm~\ref{alg:pc_app}.

\begin{algorithm}
\DontPrintSemicolon
\SetKwInOut{Input}{Input}
\SetKwInOut{Output}{Output}
\SetKwProg{Fn}{Function}{:}{\KwRet}
\SetAlgoLined
\KwIn{$V$: vertex set, $\mathcal{D}$: dataset, $T$: threshold}
\Fn{PC($V, \mathcal{D},T$)}{
$\mathcal{G}$ = complete graph on V, $ord$ = 0\;
\While{$\exists~v_i~s.t.~ |Adj(\mathcal{G}, v_i)- v_j|\geq \text{ord}$}{
    \While{$\exists$ edge $(v_i, v_j)~s.t.~|Adj(\mathcal{G}, v_i)-v_j| \geq ord$ that has not been tested}{
        select edge $(v_i, v_j)$ in $\mathcal{G}~s.t.~|Adj(\mathcal{G}, v_i)-v_j| \geq ord$\;
        \While{$\exists~S$ that has not been tested}{
            choose $S \subseteq Adj(\mathcal{G}, v_i)-v_j$, $|S|=ord$\;
            \If{$indep\_test(ij|S) \geq T$}{
                remove $(v_i, v_j)$ from $\mathcal{G}$\;
                break\;
            }
        }    
    }
    $ord$ = $ord$ + 1\;}
    Output $\mathcal{G}$\;
}
\caption{\pc Algorithm.}
\label{alg:pc_app}
\end{algorithm}

%% file: em_svt.tex
\section{Exponential Mechanism \& Sparse Vector Technique}
\label{sec:em_svt}

For completeness, we provide detailed description of exponential mechanism and exponential mechanism.

\paragraph{Exponential Mechanism.}
Exponential mechanism is designed for differentially private selection from infinite output set.
It computes an utility score for each candidate output and randomly selects from the output candidates based on probability derived from the utility score.
The pseudo-code for exponential mechanism is shown in Algorithm~\ref{alg:em}.

\begin{algorithm}
\DontPrintSemicolon
\SetKwInOut{Input}{Input}
\SetKwInOut{Output}{Output}
\SetKwProg{Fn}{Function}{:}{\KwRet}
\SetAlgoLined
\KwIn{$\mathcal{D}$: dataset, $O$: output set, $u$: 1-sensitive utility function, $\epsilon$: privacy parameters.}
\Fn{EM($\mathcal{D},O,u,\epsilon$)}{
    Initiate $U$ as an empty lists\;
    \For{$o \in O$}{
        Append $u(\mathcal{D}, o)$ to $U$\;
    }
    Randomly select $o$ from $O$ according to probability $\frac{\exp(\epsilon u_o/2)}{\sum_{u_i\in U}\exp(\epsilon u_i/2)}$.
}
\caption{Exponential Mechanism}
\label{alg:em}
\end{algorithm}

\paragraph{Sparse Vector Technique.} Sparse vector technique is a widely used differentially private mechanism. It can answer a large number of queries while only paying privacy cost for a small portion of them. The pseudo-code for sparse vector technique is shown in Algorithm~\ref{alg:svt}.

\begin{algorithm}
\DontPrintSemicolon
\SetKwInOut{Input}{Input}
\SetKwInOut{Output}{Output}
\SetKwProg{Fn}{Function}{:}{\KwRet}
\SetAlgoLined
\KwIn{$D$: dataset, $\{f_i\}$: 1-sensitive queries, $T$: threshold, \\~~~~~~~~~~~~$c$: quota of above-threshold queries, $(\epsilon,\delta)$: privacy parameters.}
\Fn{SVT($\mathcal{D},\{f_i\},T,c,\epsilon,\delta$)}{
    \leIf{$\delta=0$}
    {Let $\sigma=\frac{2c}{\epsilon}$}
    {Let $\sigma=\frac{\sqrt{32c\log{\frac{1}{\delta}}}}{\epsilon}$ }
    Let count = 0, $\hat{T}_{count}=T+Lap(\sigma)$\;
    \For{Each query $i$}{
        Let $\nu=Lap(2\sigma)$\;
        \If{$f_i(\mathcal{D})+\nu_i\geq\hat{T}_{count}$}{
            Output $a_i=\top$\;
            Let count = count + 1
            Let $\hat{T}_{count} = T+Lap(\sigma)$\;
        }
        \lElse{
            Output $a_i=\bot$
        }
        \lIf{count $\geq q$}{
            Halt
        }
    }
}
\caption{Sparse Vector Technique.}
\label{alg:svt}
\end{algorithm}

%% file: error_bound.tex
\section{Proof for Error Bound}
\label{sec:error}

\begin{thm}[Type I error bound]

Let $E_1^\alpha$ denotes the event that Algorithm~\ref{alg:onesieve} filters out $f(D)\geq T+\alpha$.
$$\mathbb{P}[E_1^\alpha] \leq \exp(-\frac{\epsilon'(\alpha+t)}{6\Delta}) - \frac{1}{4}\exp(-\frac{\epsilon'(\alpha+t)}{3\Delta})$$.
\label{thm:appendix_error}

\end{thm}

\begin{proof}

We want to upper bound the probability of $E_1^\alpha$.
Equally, we lower bound the probability of $\neg E_1^\alpha$ by
the probability that the noise on the threshold is smaller than $\frac{1}{3}(t+\alpha)$ and the noise on the query output is smaller than $\frac{2}{3}(t+\alpha)$.
Because for Laplace noise, $\mathbb{P}[x\geq w]=\exp(-w/b)$,
we have
$$
\mathbb{P}[\neg E_1^\alpha] \geq (1-\frac{1}{2}\exp(-\frac{\epsilon'(\alpha+t)}{6\Delta}))^2 = 1 - \exp(-\frac{\epsilon'(\alpha+t)}{6\Delta}) + \frac{1}{4}\exp(-\frac{\epsilon'(\alpha+t)}{3\Delta})
$$.
Thus,
$$
\mathbb{P}[E_1^\alpha] \leq 1-\mathbb{P}[\neg E_1^\alpha] \leq \exp(-\frac{\epsilon'(\alpha+t)}{6\Delta}) - \frac{1}{4}\exp(-\frac{\epsilon'(\alpha+t)}{3\Delta})
$$

\end{proof}

\begin{thm}[Type II error bound]

Let $E_2^\alpha$ denotes the event that Algorithm~\ref{alg:onesieve} fails to filter out $f(D)\leq T-\alpha$. 
If $\alpha \geq t$, then
$$\mathbb{P}[E_2^\alpha]\leq \exp(-\frac{12\epsilon\alpha+\epsilon'(\alpha-t)}{6\Delta}) - \frac{1}{4}\exp(-\frac{6\epsilon\alpha+\epsilon'(\alpha-t)}{3\Delta})$$.

\end{thm}

\begin{proof}
If $f(D)$ is not filtered out, it needs to be missed by both sparse vector technique and the Laplace mechanism.
The probability bound for being missed by the sparse vector technique is
$$
\mathbb{P}[E_{svt}^\alpha] \leq  \exp(-\frac{\epsilon'(\alpha-t)}{6\Delta}) - \frac{1}{4}\exp(-\frac{\epsilon'(\alpha-t)}{3\Delta})
$$
following similar proof path to theorem~\ref{thm:appendix_error}.
The probability being missed by the Laplace mechanism is bounded by
$$
\mathbb{P}[E_{lm}^\alpha] = \exp(-\frac{2\epsilon\alpha}{\Delta})
$$.
Thus, 
\begin{equation*}
\begin{split}
\mathbb{P}[E_2^\alpha] = \mathbb{P}[E_{svt}^\alpha]\cdot \mathbb{P}[E_{lm}^\alpha] &\leq  (\exp(-\frac{\epsilon'(\alpha-t)}{6\Delta}) - \frac{1}{4}\exp(-\frac{\epsilon'(\alpha-t)}{3\Delta}))\cdot\exp(-\frac{2\epsilon\alpha}{\Delta})\\
&= \exp(-\frac{12\epsilon\alpha+\epsilon'(\alpha-t)}{6\Delta}) - \frac{1}{4}\exp(-\frac{6\epsilon\alpha+\epsilon'(\alpha-t)}{3\Delta})
\end{split}
\end{equation*}
\end{proof}

%% file: kendall_sensitivity.tex
\section{Sensitivity of Kendall's $\tau$}
\label{sec:tau}

In this section, we derive the sensitivity of Kendall's $\tau$ and its conditional version.
We first give the complete definition of Kendall's $\tau$ and its conditional version.

\begin{definition}[Kendall's $\tau$]
Let $\{(a_1, b_1), \cdots, (a_n, b_n)\}$ denotes the observations. 
A pair of observation indices $(i, j)$ are called \emph{concordant} if $a_i>a_j$ and $b_i>b_j$. 
Otherwise $(i, j)$ is called \emph{discordant}.
Kendall's $\tau$ is defined as
$$
\tau_{ij} := \frac{2|C-D|}{n(n-1)}
$$
where $C$ is the number of concordant pairs and $D$ is the number of discordant pairs.
\label{def:tau}
\end{definition}

Kusner et al.~\cite{kusner2015private} derive the sensitivity for unconditional Kendall's $\tau$ when the neighboring relation between datasets are constrained to replacement. 
The first step towards complete sensitivity analysis for unconditional Kendall's $\tau$ is to extend the neighboring relation to increment.

\begin{thm}
Kendall's $\tau$ is $\frac{2}{n-1}$-sensitive.
\end{thm}

\begin{proof}
When the neighboring datasets are defined by replacement, the proof is done in~\cite{kusner2015private}. 
Now we prove that the sensitivity bound generalizes to neighboring datasets defined by increment.

If we increment a dataset by one row, $|C-D|$ can increase by at most $n$.

\begin{equation*}
\begin{split}
    s(\tau_{ij}) \leq \frac{|C-D|+n}{\frac{1}{2}n(n+1)} - \frac{|C-D|}{\frac{1}{2}n(n-1)} \leq \frac{|C-D|+n}{\frac{1}{2}n(n-1)} - \frac{|C-D|}{\frac{1}{2}n(n-1)}
    & \leq \frac{2}{n-1}
\end{split}
\end{equation*}

\end{proof}

\begin{thm}
If the conditional variables have $k$ blocks, then conditional Kendall's $\tau$ is $\frac{c_\tau}{\sqrt{n-1}}$-sensitive, where $c_\tau$ is an explicit constant typically close to $\frac{9}{2}$.
\end{thm}

\begin{definition}[Conditional Kendall's $\tau$]
We omit the pair indices $i, j$ and use $\tau_i$ to represent Kendall's $\tau$ in the $i$th block of the conditional variables. If there are $k$ blocks in total, then conditional Kendall's $\tau$ is defined as
$$\tau = \frac{\sum_{i=1}^{k} w_i\tau_i}{\sqrt{\sum_{j=1}^{k} w_j}}$$
where $w_i=\frac{9n_i(n_i-1)}{2(2n_i+5)}$ is the inverse of $\tau_i$'s variance.
%
\label{def:ctau}
\end{definition}

\begin{proof}
If the $i$th block contains $n_i$ observations, then $s(\tau_i)=\frac{2}{n_i-1}$.

Then we need to bound $\frac{w_i}{\sqrt{\sum_{j=1}^k w_j}}$ and its sensitivity.
Assuming $\forall i\in [1, k], n_i\geq c_1$, then $c_2(n_i-1) \leq w_i\leq \frac{9(n_i-1)}{4}$ for some explicit constants $c_2=\frac{9c_1}{2(2c_1+5)}$.
Thus 
$$\frac{w_i}{\sqrt{\sum_{j=1}^k w_j}} \leq \frac{9(n_i-1)}{4\sqrt{c_2(n-k)}}$$
and 
$$ s(\frac{w_i}{\sqrt{\sum_{j=1}^k w_j}}) \leq \frac{w_i'}{\sqrt{\sum_{j\neq i} w_j + w_i'}}-\frac{w_i}{\sqrt{\sum_{j=1}^k w_j}}\leq \frac{w_i'-w_i}{\sqrt{\sum_{j=1}^k w_j}} \leq \frac{9}{4\sqrt{c_2(n-k)}}$$.
Thus the complete sensitivity is bounded as follow.
$$
s(\tau) \leq (\frac{w_i}{\sqrt{\sum_{j=1}^k w_j}} + s(\frac{w_i}{\sqrt{\sum_{j=1}^k w_j}}))(\tau_i + s(\tau_i))-\frac{w_i}{\sqrt{\sum_{j=1}^k w_j}}\tau_i \leq \frac{27}{4\sqrt{c_2(n-k)}}+\frac{9}{2c_1\sqrt{c_2(n-k)}}
$$
\end{proof}

%% file: spearman_sensitivity.tex
\section{Sensitivity of Spearman's $\rho$}
\label{sec:rho}

In this section, we derive the sensitivity of Spearman's $\rho$ and its conditional version. We first give the complete definition of Spearman's $\rho$ and its conditional version.

\begin{definition}[Spearman's $\rho$]
Let $\{(a_1, b_1), \cdots, (a_n, b_n)\}$ denotes the observations. 
If we independently sort the observations $\{a_1, \cdots, a_n\}$ and $\{b_1, \cdots, b_n\}$ in ascending order. 
Let $d_i$ represent the distance between the order of $a_i$ and $b_i$.
Spearman's $\rho$ is defined as
$$
\rho = |1-\frac{6\sum_{i=1}^nd_i^2}{n(n^-1)}|
$$.
\label{def:rho}
\end{definition}

Kusner et al.~\cite{kusner2015private} derive the sensitivity for unconditional Spearman's $\rho$ when the neighboring relation between datasets are constrained to replacement. 
The first step towards complete sensitivity analysis for unconditional Spearman's $\rho$ is to extend the neighboring relation to increment.

\begin{thm}
Spearman's $\rho$ is $\frac{30}{n}$-sensitive.
\end{thm}

\begin{proof}
When the neighboring datasets are defined by replacement, the proof is done in~\cite{kusner2015private}. 
Now we prove that the sensitivity bound generalizes to neighboring datasets defined by increment.
And we denote the incremented observation with $(a_{n+1}, b_{n+1})$.
First, $\forall i\neq n+1$, $d_i$ changes at most 2.
Thus $d_i^2-(d_i-2)^2\leq 4(d_i-1) \leq 4(m-2)$, because $d_i$ is smaller than $m-1$.
Besides, $d_{n+1}$ is at most $m$.
Therefore, the sensitivity of $\rho$ si bounded by
$$
s(\rho) \leq \frac{30m(m-1)}{m(m^2-1)} \leq \frac{30}{m}
$$

\end{proof}

%

\begin{definition}[Conditional Spearman's $\rho$]
We omit the pair indices $i, j$ and use $\rho_i$ to represent Spearman's $\rho$ in the $i$th block of the conditional variables. If there are $k$ blocks in total, then conditional Spearman's $\rho$ is defined as
$$\rho = \frac{\sum_{i=1}^{k} w_i\rho_i}{\sqrt{\sum_{j=1}^{k} w_j}}$$
where $w_i=n_i-1$.
\label{def:crho}
\end{definition}

\begin{thm}
Conditional Spearman's $\rho$ is $\frac{c_\rho\sqrt{k}}{\sqrt{n-k}}$-sensitive, where $c_\rho$ is an explicit constant typically close to $31$.
\end{thm}

\begin{proof}
If the $i$th block contains $n_i$ observations, then $s(\rho_i)=\frac{30}{n_i}$.

Then we need to bound $\frac{w_i}{\sqrt{\sum_{j=1}^k w_j}}$ and its sensitivity.
$$\frac{w_i}{\sqrt{\sum_{j=1}^k w_j^2}} \leq \frac{n_i-1}{\sqrt{n-k}}$$
and 
$$ s(\frac{w_i}{\sqrt{\sum_{j=1}^k w_j}}) \leq \frac{w_i'}{\sqrt{\sum_{j\neq i} w_j + w_i'}}-\frac{w_i}{\sqrt{\sum_{j=1}^k w_j}}\leq \frac{w_i'-w_i}{\sqrt{\sum_{j=1}^k w_j}} \leq \frac{1}{\sqrt{n-k}}$$.
Thus the complete sensitivity is bounded as follow.
$$
s(\rho) \leq (\frac{w_i}{\sqrt{\sum_{j=1}^k w_j}} + s(\frac{w_i}{\sqrt{\sum_{j=1}^k w_j}}))(\rho_i + s(\rho_i))-\frac{w_i}{\sqrt{\sum_{j=1}^k w_j}}\rho_i \leq \frac{31}{\sqrt{n-k}}+\frac{30}{c_1\sqrt{n-k}}
$$
\end{proof}

%% file: aggregate.tex
\section{Reconcile Sensitive Independence Test}
\label{sec:aggregate}

As an attempt to reconcile independence tests with infinite sensitivity such as G-test or $\chi^2$-test in \alg, we use subsample-and-aggregate and median aggregation with local sensitivity to stabilize these independence tests.

\begin{definition}[Subsample-and-aggregate~\cite{dwork2014algorithmic}]
Let $f$ be the function of interest.
In subsample-and-aggregate, the input database is randomly partitioned into $m$ blocks and $f$ is computed exactly on each block.
The outcomes are then aggregated using a differentially private aggregation mechanism such as trimmed mean.
\end{definition}

In order to use subsample-and-aggregate, we clip the outcome of the independence test to a bounded range and estimate the median by adding noise calibrated to the smooth sensitivity~\cite{nissim2007smooth}.

\begin{definition}[Median Aggregation with Local Sensitivity~\cite{nissim2007smooth}]
Let $S_{med}(x)$ represent the smooth sensitivity of the median of a given input $x$.
$S_{med}(x)$ can be $\beta$ upper bounded by the following formula in $\mathcal{O}(n\log(n))$.
$$
S_{med}(x) = \max_{k=0, \cdots, n}(e^{-k\epsilon}\cdot \max_{t=0, \cdots, k+1}(x_{m+t}-x_{m+t-k-1}))
$$
where $m$ is the median index.
Let $Z$ be a random value taken from an $(\alpha, \beta)$-admissible noise probability density function, then $med(x)+\frac{S_{med}(x)}{\alpha}\cdot Z$ is ($\epsilon, \delta$)-differentially private where $\epsilon$ and $\delta$ depends on $\alpha$ and $\beta$.
For instance, the Laplace distribution $p(z)=\frac{1}{2}\cdot e^{-|z|}$ is $(\epsilon/2, \epsilon\ln(1/\delta)/2)$-admissible; the Gaussian distribution $p(z)=\frac{1}{2\pi}\cdot e^{-z^2/2}$ is $(\epsilon/\sqrt{\ln(1/\delta)}, \epsilon/2\ln(1/\delta))$-admissible.
\end{definition}

\begin{algorithm}
\DontPrintSemicolon
\SetKwInOut{Input}{Input}
\SetKwInOut{Output}{Output}
\SetKwProg{Fn}{Function}{:}{\KwRet}
\SetAlgoLined
\KwIn{$D$: dataset, $m$: number of blocks, $f$: independence test function, $T$: threshold, $\epsilon, \delta$: privacy parameters, $Z$: $(\alpha, \beta)$-admissible distribution.}
\Fn{ReconciledIDT($D, f, T, \epsilon, \delta, Z$)}{
Partition the dataset in $m$ blocks $D_1, \cdots, D_m$.\;
Compute $f(D_1), \cdots, f(D_m)$.\;
Let $z\leftarrow Z$.\;
Output $med({f(D_1), \cdots, f(D_m)}) + \frac{S_{med}({f(D_1), \cdots, f(D_m)})}{\alpha}\cdot z$\;
}
\caption{Reconciled Independence Tests.}
\label{alg:saa}
\end{algorithm}

%% file: main.bbl
\begin{thebibliography}{10}

\bibitem{abadi2016deep}
Martin Abadi, Andy Chu, Ian Goodfellow, H~Brendan McMahan, Ilya Mironov, Kunal
  Talwar, and Li~Zhang.
\newblock Deep learning with differential privacy.
\newblock In {\em Proceedings of the 2016 ACM SIGSAC Conference on Computer and
  Communications Security}, pages 308--318, 2016.

\bibitem{balle2018privacy}
Borja Balle, Gilles Barthe, and Marco Gaboardi.
\newblock Privacy amplification by subsampling: Tight analyses via couplings
  and divergences.
\newblock In {\em Advances in Neural Information Processing Systems}, pages
  6277--6287, 2018.

\bibitem{bernstein2018differentially}
Garrett Bernstein and Daniel~R Sheldon.
\newblock Differentially private bayesian inference for exponential families.
\newblock In {\em Advances in Neural Information Processing Systems}, pages
  2919--2929, 2018.

\bibitem{chickering2002optimal}
David~Maxwell Chickering.
\newblock Optimal structure identification with greedy search.
\newblock {\em Journal of machine learning research}, 3(Nov):507--554, 2002.

\bibitem{dimitrakakis2014robust}
Christos Dimitrakakis, Blaine Nelson, Aikaterini Mitrokotsa, and Benjamin~IP
  Rubinstein.
\newblock Robust and private bayesian inference.
\newblock In {\em International Conference on Algorithmic Learning Theory},
  pages 291--305. Springer, 2014.

\bibitem{dwork2006calibrating}
Cynthia Dwork, Frank McSherry, Kobbi Nissim, and Adam Smith.
\newblock Calibrating noise to sensitivity in private data analysis.
\newblock In {\em Theory of cryptography conference}, pages 265--284. Springer,
  2006.

\bibitem{dwork2014algorithmic}
Cynthia Dwork, Aaron Roth, et~al.
\newblock The algorithmic foundations of differential privacy.
\newblock {\em Foundations and Trends{\textregistered} in Theoretical Computer
  Science}, 9(3--4):211--407, 2014.

\bibitem{glymour2019review}
Clark Glymour, Kun Zhang, and Peter Spirtes.
\newblock Review of causal discovery methods based on graphical models.
\newblock {\em Frontiers in Genetics}, 10, 2019.

\bibitem{heikkila2017differentially}
Mikko Heikkil{\"a}, Eemil Lagerspetz, Samuel Kaski, Kana Shimizu, Sasu Tarkoma,
  and Antti Honkela.
\newblock Differentially private bayesian learning on distributed data.
\newblock In {\em Advances in neural information processing systems}, pages
  3226--3235, 2017.

\bibitem{hoyer2009nonlinear}
Patrik~O Hoyer, Dominik Janzing, Joris~M Mooij, Jonas Peters, and Bernhard
  Sch{\"o}lkopf.
\newblock Nonlinear causal discovery with additive noise models.
\newblock In {\em Advances in neural information processing systems}, pages
  689--696, 2009.

\bibitem{johnson2013privacy}
Aaron Johnson and Vitaly Shmatikov.
\newblock Privacy-preserving data exploration in genome-wide association
  studies.
\newblock In {\em Proceedings of the 19th ACM SIGKDD international conference
  on Knowledge discovery and data mining}, pages 1079--1087, 2013.

\bibitem{kendall1938new}
Maurice~G Kendall.
\newblock A new measure of rank correlation.
\newblock {\em Biometrika}, 30(1/2):81--93, 1938.

\bibitem{korb2010bayesian}
Kevin~B Korb and Ann~E Nicholson.
\newblock {\em Bayesian artificial intelligence}.
\newblock CRC press, 2010.

\bibitem{korn1984kendall}
Edward~L Korn.
\newblock Kendall's tau with a blocking variable.
\newblock {\em Biometrics}, pages 209--214, 1984.

\bibitem{kusner2015private}
Matt~J Kusner, Yu~Sun, Karthik Sridharan, and Kilian~Q Weinberger.
\newblock Private causal inference.
\newblock {\em arXiv preprint arXiv:1512.05469}, 2015.

\bibitem{lauritzen1988local}
Steffen~L Lauritzen and David~J Spiegelhalter.
\newblock Local computations with probabilities on graphical structures and
  their application to expert systems.
\newblock {\em Journal of the Royal Statistical Society: Series B
  (Methodological)}, 50(2):157--194, 1988.

\bibitem{lyu2016understanding}
Min Lyu, Dong Su, and Ninghui Li.
\newblock Understanding the sparse vector technique for differential privacy.
\newblock {\em arXiv preprint arXiv:1603.01699}, 2016.

\bibitem{mcdonald2009handbook}
John~H McDonald.
\newblock {\em Handbook of biological statistics}, volume~2.
\newblock sparky house publishing Baltimore, MD, 2009.

\bibitem{mchugh2013chi}
Mary~L McHugh.
\newblock The chi-square test of independence.
\newblock {\em Biochemia medica: Biochemia medica}, 23(2):143--149, 2013.

\bibitem{mironov2017renyi}
Ilya Mironov.
\newblock R{\'e}nyi differential privacy.
\newblock In {\em 2017 IEEE 30th Computer Security Foundations Symposium
  (CSF)}, pages 263--275. IEEE, 2017.

\bibitem{nissim2007smooth}
Kobbi Nissim, Sofya Raskhodnikova, and Adam Smith.
\newblock Smooth sensitivity and sampling in private data analysis.
\newblock In {\em Proceedings of the thirty-ninth annual ACM symposium on
  Theory of computing}, pages 75--84, 2007.

\bibitem{nocedal2006numerical}
Jorge Nocedal and Stephen Wright.
\newblock {\em Numerical optimization}.
\newblock Springer Science \& Business Media, 2006.

\bibitem{pearl2009causal}
Judea Pearl et~al.
\newblock Causal inference in statistics: An overview.
\newblock {\em Statistics surveys}, 3:96--146, 2009.

\bibitem{scutari2014bayesian}
Marco Scutari and Jean-Baptiste Denis.
\newblock {\em Bayesian networks: with examples in R}.
\newblock CRC press, 2014.

\bibitem{shimizu2011directlingam}
Shohei Shimizu, Takanori Inazumi, Yasuhiro Sogawa, Aapo Hyv{\"a}rinen,
  Yoshinobu Kawahara, Takashi Washio, Patrik~O Hoyer, and Kenneth Bollen.
\newblock Directlingam: A direct method for learning a linear non-gaussian
  structural equation model.
\newblock {\em Journal of Machine Learning Research}, 12(Apr):1225--1248, 2011.

\bibitem{spearman1961proof}
Charles Spearman.
\newblock The proof and measurement of association between two things.
\newblock 1961.

\bibitem{spirtes2000causation}
Peter Spirtes, Clark~N Glymour, Richard Scheines, and David Heckerman.
\newblock {\em Causation, prediction, and search}.
\newblock MIT press, 2000.

\bibitem{taylor1987kendall}
Jeremy~MG Taylor.
\newblock Kendall's and spearman's correlation coefficients in the presence of
  a blocking variable.
\newblock {\em Biometrics}, pages 409--416, 1987.

\bibitem{wang2018subsampled}
Yu-Xiang Wang, Borja Balle, and Shiva Kasiviswanathan.
\newblock Subsampled r$\backslash$'enyi differential privacy and analytical
  moments accountant.
\newblock {\em arXiv preprint arXiv:1808.00087}, 2018.

\bibitem{xu2017differential}
Depeng Xu, Shuhan Yuan, and Xintao Wu.
\newblock Differential privacy preserving causal graph discovery.
\newblock In {\em 2017 IEEE Symposium on Privacy-Aware Computing (PAC)}, pages
  60--71. IEEE, 2017.

\bibitem{zhang2006extensions}
Kun Zhang and Lai-Wan Chan.
\newblock Extensions of ica for causality discovery in the hong kong stock
  market.
\newblock In {\em International Conference on Neural Information Processing},
  pages 400--409. Springer, 2006.

\bibitem{zhang2009ica}
Kun Zhang, Heng Peng, Laiwan Chan, and Aapo Hyv{\"a}rinen.
\newblock Ica with sparse connections: Revisited.
\newblock In {\em International Conference on Independent Component Analysis
  and Signal Separation}, pages 195--202. Springer, 2009.

\end{thebibliography}
